\documentclass{llncs}
\usepackage{amsmath, amssymb}
\usepackage{graphicx}
\usepackage{dsfont}
\usepackage{url}
\usepackage{hyperref,setspace}
\usepackage{epsfig,psfrag}
\usepackage{enumerate}
\usepackage{multirow}
\usepackage{mathpazo}
\newcommand{\bs}{\boldsymbol}
\renewcommand{\note}[1]{}
\renewcommand{\note}[1]{~\\\frame{\begin{minipage}[c]{\textwidth}\vspace{2pt}\center{#1}\vspace{2pt}\end{minipage}}\vspace{3pt}\\}

\newcommand{\etal}{\emph{et al.}}
\long\def\symbolfootnote[#1]#2{\begingroup%
\def\thefootnote{\fnsymbol{footnote}}\footnote[#1]{#2}\endgroup}

\newcommand{\One}{\mathbf{1}}
\newcommand{\Real}{\mathds{R}}
\title{Bilinear Games: Polynomial Time Algorithms for Rank Based Subclasses}
\author{Jugal Garg\inst{1}\thanks{Work done while the author was an intern at Microsoft Research India} \and Albert Xin Jiang\inst{2} \and Ruta Mehta\inst{1}}
\institute{Indian Institute of Technology, Bombay\\ \email{jugal,ruta@cse.iitb.ac.in}\and University of British Columbia\\ 
\email{jiang@cs.ubc.ca}}
\begin{document}
\maketitle

\begin{abstract}
Motivated by the sequence form formulation of Koller et al. \cite{KollerMegiddoVonStengel1996}, this paper defines {\em
bilinear games}, and proposes efficient algorithms for its rank based subclasses. 
{\em Bilinear games} are two-player non-cooperative single-shot games with compact polytopal strategy sets and two payoff
matrices $(A,B)$ such that when $(x,y)$ is the played strategy profile, the payoffs of the players are $x^TAy$ and $x^TBy$
respectively. 
We show that bilinear games are very general and capture many interesting classes of games like bimatrix games, two player Bayesian games,
polymatrix games, two-player extensive form games with perfect recall etc. as special cases, and hence are hard to solve in general. 

$\ \ \ \ $ Existence of a (symmetric) Nash equilibrium for (symmetric) bilinear games follow directly from the known results. 
For a given bilinear game, we define its {\em Best Response Polytopes} (BRPs) and characterize the set of Nash equilibria as {\em fully-labeled} pairs in the BRPs.
We consider a rank based hierarchy of bilinear games, where rank of a game $(A,B)$ is defined as $rank(A+B)$. 
In this paper, we give polynomial time algorithms to compute Nash equilibrium for special classes of bilinear games:
\begin{itemize}
\item Rank-$1$ games ({\em i.e.,} $rank(A+B)=1$).
\item FPTAS for constant rank games ({\em i.e.,} $rank(A+B)$ is constant).
\item When $rank(A)$ or $rank(B)$ is constant. This improves the results by Lipton et al. \cite{lipton2003simple} and
Kannan et al. \cite{kannan09rank}, for bimatrix games with low rank matrices.
\end{itemize}
\end{abstract}

\section{Introduction}
In the last decade, there has been much research at the interface of computer science and game theory (see e.g.
\cite{AGTBook,MAS}).  One fundamental class of computational problems in game theory is the computation of \emph{solution
concepts} of finite games. Nash \cite{Nash1951} proved that in any finite game there always exists a steady state, from
where no player gains by unilateral deviation. Such a steady state has since been named
Nash equilibrium (NE)  and is perhaps the most well-known and well-studied game-theoretic solution concept.
However, computing a Nash equilibrium is nontrivial, and indeed the recent series of papers
\cite{ChenDeng06,DaskalakisGoldbergPapa06,GoldbergPapa06} established that the problem is PPAD-complete for finite games in the
standard \emph{normal form} representation, even for games with only two players.
Furthermore it is PPAD-complete to
even find a $\frac{1}{n^{\theta(1)}}$-approximate Nash equilibrium \cite{Chen06approx}.
In light of these negative results, one direction is to identify subclasses of games for which the problem is tractable.

A two-player normal form game can be represented by two payoff matrices, say $A$ and $B$, one for each player, and hence is
also known as bimatrix game \cite{AGTBook}. 
For bimatrix games, polynomial time NE computation algorithms are known for many subclasses, including zero-sum games \cite{dan},  (quasi-)
concave games \cite{ks}, and games with low rank payoff matrices \cite{lipton2003simple}. A line of work focuses on games of
small \emph{rank}, defined as $rank(A+B)$ by Kannan and Theobald \cite{kannan09rank}. They gave a fully polynomial time
approximation scheme (FPTAS) for fixed rank games and recently Adsul et al. \cite{agms} gave a polynomial time algorithm for
computing an exact Nash equilibrium for rank-$1$ games.
Specifying the two payoff matrices of a bimatrix game requires a polynomial number of entries in the numbers of pure strategies available to
the players. This is adequate when the set of pure strategies are explicitly given. However, there are situations where the natural
description gives the set of pure strategies implicitly, and as a result they may be exponential in the description of the game.
For example, normal form (bimatrix) representation of two player extensive-form game may have exponentially many strategies in
the size of the extensive-form description \cite{FudTir1991}. In such a case, even if the resulting bimatrix game has a fixed rank, the above results may not be
applied for efficient computation.

Nevertheless, certain types of extensive-form games have some combinatorial structure which can be exploited.
Koller, Megiddo and von Stengel \cite{KollerMegiddoVonStengel1996} converted an arbitrary two-player, perfect-recall, extensive form game
into a payoff-equivalent two-player game with continuous strategy sets. In this derived formulation, which they call the \emph{sequence
form}, there is a pair of payoff matrices $A$ and $B$, one for each player. Further, their strategy sets turn out to be compact polytopes
in Euclidean space of polynomial dimension. Given a pair of strategies $(x,y)$, utilities of the players are $x^TAy$ and $x^TBy$
respectively. Interestingly, the sequence form requires only a polynomial number of bits to specify.

Motivated by the sequence form of Koller \etal, we define {\em bilinear games}, which are two-player, non-cooperative,
single shot games represented by two payoff matrices, say $A$ and $B$, of dimension $M\times N$ and two polytopal compact
strategy sets $X=\{x \in \Real^M\ |\ Ex=e,\ x\ge0\}$ and $Y=\{y \in \Real^N\ |\ Fy=f,\ y\ge 0\}$. 
If $(x,y) \in X\times Y$ is the played strategy profile, then $x^TAy$ and $x^TBy$ are the utilities derived by player one and player two
respectively. In other words, the payoffs are bilinear functions of strategies, hence the name {\em bilinear games}. 
The scope of bilinear games is large enough to capture many interesting classes of games besides two-player extensive form
games with perfect recall. 
For example, for two-player Bayesian games
\cite{howson1974bayesian,Ponssard80lp}, polymatrix games \cite{howson1972equilibria}, and various classes of optimization duels
\cite{stoc11Dueling}, researchers have proposed polynomial-sized payoff-equivalent formulations which (either explicitly or implicitly)
turn out to be  bilinear games (see Section \ref{examples} for details).
Intuitively the polytopal strategy sets are concise representations of the original sets of mixed strategies as marginal probabilities, and
$x^TAy$ and $x^TBy$ express the expected utilities of the original game in terms of these marginal probabilities. 

\begin{remark}
Note that our formulation can express arbitrary polytopes as strategy space, for example if the strategy set of a player is expressed as
$\{x: Gx\leq g\}$, {\em i.e.}, the intersection of a set of half-spaces, it can be transformed to an equivalent game with strategy set
$\{x': Ex'=e, x'\geq 0\}$ using standard techniques ({\em i.e.}, by adding slack variables, substituting unbounded $x_i$ with
$x_i^+ - x_i^- $, $x_i^+,x_i^- \geq 0$, and modifying the payoff matrices accordingly).
\end{remark}

As we have seen that many different types of games can be concisely described as bilinear games, 
designing efficient algorithms for the more general bilinear games seem to be important as well as challenging.
Since bimatrix games is a subclass of bilinear games (see Section \ref{examples}), all the hardness results of bimatrix
games automatically apply to bilinear games as well. Therefore, the only hope is to design efficient algorithms or FPTAS
for the special subclasses. There are many similarities between bilinear and bimatrix games, for example, payoffs are
represented by two matrices, and utilities are bilinear functions of strategy vectors,
hence it is natural to try to adapt algorithms for bimatrix games to bilinear games. However, a technical challenge is that the polytopal strategy
sets of bilinear games are generally much more complex than the bimatrix case; in particular the number of vertices may be
exponential, while the set of mixed strategies is just a simplex. Recently, we were pointed to the constrained games, similar to bilinear games, considered by
Charnes \cite{Charnes53}. The linear programming technique by Charnes \cite{Charnes53} also works for zero-sum bilinear
games, {\em i.e.,} those with $(A+B)=0$ (see also \cite{koller1994} for sequence form, \cite{stoc11Dueling} for another
derivation,  \cite{Ponssard80lp} for zero-sum Bayesian games and \cite{Daskalakis09minmax} for zero-sum polymatrix games).
Further, it is easy to show that the linear complementarity program (LCP) characterization for the set of NE of a sequence
form game \cite{KollerMegiddoVonStengel1996} works for bilinear games as well. However the Lemke-Howson algorithm
may not be directly applied to the general bilinear games. There are results for certain subclasses of bilinear games with
specific structure in their strategy sets, e.g., Howson and Rosenthal \cite{howson1974bayesian} adapted the Lemke-Howson
algorithm to two-player Bayesian games and Koller \etal \cite{KollerMegiddoVonStengel1996} adapted Lemke's algorithm
\cite{Lemke1965} to two-player extensive-form games.
\\ \\
\noindent{\bf Our Contribution. }
In Section \ref{gne}, we define the bilinear game and show that the existence of a (symmetric) Nash equilibrium in a
(symmetric) bilinear game directly follows from the known results. Note that given a strategy of a player, the other
player would like to play a utility maximizing strategy. 
We formulate this problem as a primal-dual LP, similar to the Koller et al. formulation.
Using the complementarity conditions of the primal-dual, we characterize the Nash equilibria and then define {\em Best
Response Polytopes} (BRPs) and the notion of {\em fully-labeled} pairs in BRPs. Further, we show one-to-one correspondence
between the Nash equilibria and fully-labeled pairs. This in turn gives a quadratic programming (QP) formulation for the NE
computation problem. 

Next, we extend Kannan and Theobald's \cite{kannan09rank} rank-based hierarchy for bimatrix games to bilinear games, by
defining the rank of a bilinear game with payoff matrices $(A,B)$ as the rank of $(A+B)$. Zero-sum games are rank-$0$ games
and a NE for these games can be computed efficiently, as discussed above. 
In Section \ref{rank1}, we show that in spite of a very general structure of the strategy sets in bilinear games, the basic
approach given by Adsul et al. \cite{agms} to compute a NE of a rank-$1$ bimatrix game, can be generalized to compute a NE
of a rank-$1$ bilinear game by solving a rank-1 QP. While solving a rank-1 general QP is NP-hard \cite{ravi}, those
arising from the bilinear games can be solved in polynomial time.
In Section \ref{fptas}, we discuss two FPTAS algorithms for the fixed rank bilinear games, which are generalization of the
algorithms by Kannan and
Theobald \cite{kannan09rank} for the bimatrix games. Finally, in Section \ref{slowrank}, using the structure of BRPs, we
obtain a polynomial time algorithm for the case when the rank of either $A$ or $B$ is a constant, and rank of $E$
and $F$ are also constant. Since, a bimatrix game can be thought of as a bilinear game with $E$ and $F$ being a single row
of $1$s (see Example \ref{bimatrix} of Section \ref{examples}), this algorithm improves upon a result by Lipton et al.
\cite{lipton2003simple} and Kannan et al. \cite{kannan09rank} for
bimatrix games, where they require the rank of both $A$ and $B$ to be constant. This approach also gives an enumeration
algorithm for extreme equilibria, which runs in polynomial time under the above assumption and exponential time for the
general bilinear games. 

The following table summarizes all the NE computation results while keeping the bimatrix games in
perspective.

\begin{center}
\begin{tabular}{|c|c|c|}\hline
{\bf Results} & {\bf Bimatrix games} & {\bf Bilinear games} \\\hline \hline
Existence of (symmetric) NE & Nash  \cite{Nash1951} & Easy to show using (\cite{Nash1951}) \cite{Glicksberg1952} \\ \hline
LCP formulation for NE & Known \cite{AGTBook} & Koller et al. \cite{KollerMegiddoVonStengel1996} \\ \hline
NE as fully-labeled pairs of BRPs & Known \cite{agt} & This paper \\ \hline
Zero-sum games & Linear programming \cite{dan} & Linear programming \cite{koller1994} \\ \hline
Rank-$1$ games & Adsul et al. \cite{agms} & This paper \\ \hline
FPTAS for fixed rank games & Kannan and Theobald \cite{kannan09rank} & This paper \\ \hline
Games with low rank matrices & Lipton et al. \cite{lipton2003simple} & This paper \\ \hline
\end{tabular}
\end{center}

\section{Bilinear Games and Nash Equilibria}\label{gne}
{\bf Notations.} We consider a vector $x$ as a column vector by default and for the row vector, we use transpose 
({\em i.e.,} $x^T$). A ``$0$" in the block representation of a matrix, is the matrix with all zero entries of appropriate
dimension, and ``$1_k$" is a vector of all $1$s of length $k$. 
Let $x \in \Real^n$ be a vector and $c \in \Real$ be a scalar, then by $x\le c$ we mean, $\forall i\le n,\ x_i\le c$.
For a given matrix $X$, $X_i$ denotes the $i^{th}$ row of $X$, $X^{j}$ denotes the $j^{th}$ column of $X$ and $|X|$ denotes the maximum absolute entry in $X$, {\em i.e.,}
$|X|=\max_{ij} |X_{ij}|$.  For a set $S$, $\Delta(S)$ 
denotes the set of probability distribution vectors over the elements of $S$, {\em i.e.,} $\Delta(S)=\{x \in \Real^{|S|}\ |\
x\ge0,\ \sum_{i \in S} x_i=1\}$. 

Bilinear games are two-player non-cooperative, single shot games. A bilinear game is represented by two $M\times N$
dimensional payoff matrices $A$ and $B$, one for each player, and two compact polytopal strategy sets. 
Let $S_1=\{1,\dots,M\}$ be the set of rows and $S_2=\{1,\dots,N\}$ be the set of columns of the matrices. Let $E \in
\Real^{k_1\times M}$ and $F \in \Real^{k_2\times N}$ be the matrices, and $e \in \Real^{k_1}$ and $f \in \Real^{k_2}$ be the
vectors. The strategy set of the first-player is $X=\{x \in \Real^M\ |\ Ex=e, x\ge 0\}$ and the second-player is $Y=\{y \in
\Real^N\ |\ Fy=f, y\ge 0\}$. Sets $X$ and $Y$ are assumed to be compact. 
From a strategy profile $(x,y) \in X \times Y$, the payoffs obtained by the first and the second player are $x^TAy$ and $x^TBy$ 
respectively. 

From a Nash equilibrium (NE) strategy profile, no player gains by unilateral deviation. Formally,
\begin{definition}
A strategy profile $(x,y) \in X\times Y$ is a NE of the game $(A,B)$ iff $x^TAy \ge x'^TAy,\ \forall x' \in X$ and $x^TBy
\ge x^TBy',\ \forall y' \in Y$.
\end{definition}

As a direct corollary of Glicksberg's \cite{Glicksberg1952} result that there always exists a Nash equilibrium in a game whose players'
strategy spaces are convex and compact, and whose utility function for each player $i$ is continuous in all players' strategies and
quasi-concave in $i$'s strategy, we have

\begin{proposition}\label{existence}
Every bilinear game has at least one Nash equilibrium.
\end{proposition}

A bilinear game is completely represented by a six-tuple $(A,B,E,F,e,f)$ in general. However, for ease of notation we
represent it by $(A,B)$ fixing $(E,F,e,f)$. Given a strategy $y \in Y$ of the second-player, the objective of the first
player is to play $x \in X$ such that $x^T(Ay)$ is maximized, {\em i.e.,} solve the following linear program
\cite{KollerMegiddoVonStengel1996}.
\begin{eqnarray}\label{pd}
\begin{array}{ccc}
\begin{array}{ll}
max:& x^T(Ay) \\
s.t. & Ex=e \\
& x\ge0
\end{array}
\hspace{1cm}& \underrightarrow{\mbox{Dual}} &\hspace{1cm}
\begin{array}{ll}
min: & e^Tp \\
s.t. & E^Tp \ge Ay
\end{array}
\end{array}
\end{eqnarray}

Note that $p_i$ is the dual variable of the equation $E_ix=e_i$ in the above program. 
At the optimal point $(x,p)$ of (\ref{pd}), we get $x_i>0 \Rightarrow A_iy=p^T E^i,\ \forall i \in S_1$ from the
complementarity. A similar condition can be obtained for the second-player, given an $x \in X$. At a Nash equilibrium both
the conditions are satisfied, and these characterize the NE strategies as follows:
A strategy pair $(x,y) \in X\times Y$ is a Nash equilibrium of the game $(A,B)$ iff it satisfies the following conditions.
\begin{eqnarray}\label{eq1}
\begin{array}{llllll}
\exists p \in \Real^{k_1}& \mbox{ s.t. }& Ay \leq E^Tp& \mbox{ and }& \forall i \in S_1,\hspace{.06in} x_i>0 \ \ \Rightarrow\ \ & A_iy = p^TE^i \\
\exists q \in \Real^{k_2}& \mbox{ s.t. }& x^TB \leq q^TF& \mbox{ and }& \forall j \in S_2,\hspace{.06in} y_j>0 \ \ \Rightarrow\ \ & x^TB^j = q^TF^j \\
\end{array}
\end{eqnarray}

The above characterization implies that, a player plays a strategy with non-zero probability only if
it gives the maximum payoff with respect to (w.r.t.) the opponent's strategy in {\em some sense}. Such strategies are called
the {\em best response} strategies (w.r.t. the opponent's strategy). Using this fact, we define {\em best response
polytopes} (BRPs), similar to the best response polytopes of a bimatrix game \cite{agt}.

In the following expression, $x$, $y$, $p$ and $q$ are vector variables.
\begin{eqnarray}\label{eq2}
\begin{array}{llclcll}
P=\{&(y,p)\in \Real^{N+k_1}\ \hspace{3pt}|& A_iy-p^TE^i\leq 0,& \forall i \in S_1;& y_j\geq 0,& \forall j \in S_2;& \hspace{.06in}Fy=f\}\\
Q=\{&(x,q)\in \Real^{M+k_2}\ |& x_i\geq 0,& \forall i \in S_1;&\hspace{.06in} x^TB^j-q^TF^j\leq 0,& \forall j \in S_2;& \hspace{.06in}Ex=e\}
\end{array}
\end{eqnarray}

The polytope $P$ in (\ref{eq2}) is closely related to the best
response strategies of the first-player for any given strategy of the second-player and it is called the {\em best
response} polytope of the first-player. Similarly $Q$ is called the best response polytope of the
second-player.
Note that, in both the polytopes the first
set of inequalities corresponds to the first-player, and the second set corresponds to the second-player. Since $|S_1|=M$ and
$|S_2|=N$, let the inequalities be numbered from $1$ to $M$, and $M+1$ to $M+N$ in both the polytopes. Let the {\em label}
$L(v)$ of a point $v$ in the polytope be the set of indices of the tight inequalities at $v$. If a pair $(v,w) \in P \times
Q$ is such that $L(v)\cup L(w)=\{1,\dots,M+N\}$, then it is called a {\em fully-labeled pair}. The proof of the next lemma follows
using (\ref{eq1}). 

\begin{lemma}\label{ne_le}
A strategy profile $(x,y)$ is a NE of the game $(A,B)$ iff $((y,p),(x,q))\in P\times Q$ is a fully-labeled pair,
for some $p$ and $q$. 
\end{lemma}

A game is called non-degenerate if both the polytopes are non-degenerate. Note that
a fully-labeled pair of a non-degenerate game has to be a vertex-pair. Lemma \ref{ne_le} implies that a NE strategy profile has to satisfy
the following linear complementarity conditions (LCP) over $P\times Q$ \cite{KollerMegiddoVonStengel1996}.
\begin{eqnarray}\label{eq3}
((y,p),(x,q))\in P\times Q \mbox{ corresponds to a NE }\  \ \Leftrightarrow\  \ x^T(Ay-E^Tp)=0\
\mbox{ and }\hspace{-0.9cm}\nonumber \\  (x^TB-q^TF)y=0 
\end{eqnarray}

Clearly, $x^T(Ay-E^Tp)\le0 \mbox{ and } (x^TB-q^TF)y\le0$ over $P\times Q$ and hence $x^T(Ay-E^Tp)+(x^TB-q^TF)y\le0$.
Simplifying the expression using $Ex=e$ and $Fy=f$ we get $x^T(A+B)y-e^Tp-f^Tq\leq 0$ over $P\times Q$ and equality holds
iff $(x,y)$ is a NE (using (\ref{eq3})). This gives the following QP formulation which captures all the NE of
game $(A,B)$ at its optimal points.
\begin{eqnarray}\label{eq4}
\begin{array}{rl}
\mbox{max:} & x^T(A+B)y-e^Tp-f^Tq \\
&\mbox{s.t. } \ \ ((y,p),(x,q))\in P\times Q
\end{array}
\end{eqnarray}

\vspace{0.3cm}
\noindent{\bf Symmetric Bilinear Games. }Nash \cite{Nash1951} proved that any symmetric finite game has a symmetric Nash equilibrium.
The concept of symmetry can be straightforwardly adapted to the bilinear games: We say a bilinear game is symmetric if
$B=A^T$, $E=F$, and $e=f$. A strategy profile $(x,y)$ is symmetric if $x=y$. A straightforward adaptation of Nash's
\cite{Nash1951} proof yields the following proposition.

\begin{proposition}
Any symmetric bilinear game has a symmetric NE.
\end{proposition}

Note that for a symmetric game $(A,A^T)$, strategy sets $X$ and $Y$ are the same. 
From (\ref{eq1}) and (\ref{eq3}), it is clear that a symmetric NE $x\in X$ must satisfy 
$
q=p,\ Ax\le E^Tp,\ x^T(Ax - E^Tp)=0
$.
This gives the following QP formulation to capture all symmetric NE of a symmetric game $(A,E,e)$.
\begin{eqnarray}\label{symQP}
\begin{array}{rl}
\mbox{max:} & x^TAx-e^Tp \\
\mbox{s.t. } & Ax \le E^Tp;\ \ Ex=e;\ \ x\ge0 
\end{array}
\end{eqnarray}

A bilinear game $(A,B,E,F,e,f)$ can be converted to an equivalent symmetric game $(A',E',e')$, where 
\[
A'=\left[\begin{array}{cc}0& A\\
B^T & 0\end{array}\right],\ \ \
E'=\left[\begin{array}{cc}E & 0 \\
0 & F\end{array}\right],\ \ \  e'=\left[\begin{array}{c}e \\ f \end{array}\right], \mbox{ with strategy vector }z=\left[\begin{array}{c} x\\
y\end{array}\right]
\]

It is easy to check that any symmetric NE of the derived game corresponds to a NE of the original game and vice-versa.
In the next section, we discuss reductions of different games to polynomial size bilinear games, which do not seem to be
possible with the bimatrix games. 

\subsection{Examples of Bilinear Games}\label{examples}
The simplest subclass of bilinear games is the set of two-player normal-form games (bimatrix games).
\begin{example}\label{bimatrix}{\em Bimatrix Games}\\
A bimatrix game $(A,B)$ with $A,B\in \Real^{M\times N}$ can be straightforwardly transformed to the bilinear game
$(A,B,E,F,e,f)$ where $E^T=\One_M$,
$e=1$ and similarly $F^T=\One_N$, $f=1$. \qed
\end{example}

Many other interesting classes of finite games may be formulated as bilinear games. In this section, we provide a few examples
where the bilinear formulation are exponentially smaller than a direct bimatrix formulation.
\begin{example}{\em Two-player Bayesian Games}\label{ex:bayesian}\\
In a Bayesian game \cite{harsanyi1967}, there is a {\em type set} associated with each player, which is her private
information. The nature draws the type for each player from a joint distribution, which is a common knowledge, and each player gets to know
only her own type before choosing an action. The final utilities of the game is determined by types of all the players, and hence are
uncertain. 

Here we consider the two-player case, where $T_i$s are the type sets and $S_i$s are the strategy sets. 
The joint probability distribution is denoted by $p_{ts}$ for the type profile $(t,s) \in T_1\times T_2$. Let $S=S_1\times
S_2, T=T_1\times T_2,\ |T_i|=t_i$ and $|S_i|=m_i$. The utilities are the functions of actions and types, {\em i.e.,}
$u_i:S\times T\rightarrow \Real$, hence for every type profile they can be represented by the two matrices. For a type profile
$(t,s)$ let $A^{ts}$ and $B^{ts}$ denote the respective $m_1\times m_2$ dimensional payoff matrices. The strategy of a
player is to decide her play for each of her type so that her expected payoff is maximized, {\em i.e.,} $x:T_1 \rightarrow
\Delta(S_1)$ for player one and $y:T_2 \rightarrow \Delta(S_2)$ for player two. 
For a $t \in T_1$, let $x^t$ denote the mixed strategy given type $t$. 

The \emph{induced normal form} of this game is a bimatrix game in which each pure strategy of a player prescribes an action
for each of her types. Thus the size of the induced normal form is exponential in the number of types.
However, it can be formulated as a polynomial sized bilinear game as follows 

\[
\begin{array}{cccc}
A = \left[ \begin{array}{ccc}
p_{11}A^{11}&\cdots&p_{1t_2}A^{1t_2}\\
\vdots & \ddots& \vdots\\
p_{t_1 1} A^{t_1 1}& \cdots &p_{t_1 t_2}A^{t_1t_2} \end{array}\right],
&
B = \left[ \begin{array}{ccc}
p_{11}B^{11}&\cdots&p_{1t_2}B^{1t_2}\\
\vdots & \ddots& \vdots \\
p_{t_1 1} B^{t_1 1}& \cdots &p_{t_1 t_2}B^{t_1t_2} \end{array}\right],
&

E=\left[\begin{array}{ccc}
\One_m^T &0 &\cdots\\
0&\One_m^T& \\
\vdots& &\ddots
\end{array}\right],
&
F=\left[\begin{array}{ccc}
\One_n^T &0 &\cdots\\
0&\One_n^T& \\
\vdots& &\ddots
\end{array}\right]
\end{array}
\]
and $e=\One_{t_1}$, $f=\One_{t_2}$.
Given mixed strategies $x^1,\ldots, x^{t_1}, y^1,\ldots ,y^{t_2}$ of the Bayesian game, define $x=[x^{1^T}, \cdots, x^{t_1^T}]^T$ and
$y=[y^{1^T}, \cdots, y^{t_2^T}]^T$.
Then $x^TAy$ and $x^TBy$ are exactly the expected utilities of the Bayesian game.
This transformation is implicit in Howson and Rosenthal's \cite{howson1974bayesian} adaptation of Lemke-Howson algorithm to
two-player Bayesian games.  \qed 
\end{example}

\begin{example}{\em Polymatrix Games \cite{howson1972equilibria}}\\
A polymatrix game is an $n$-player game in which each player's utility is the sum of the utilities resulting from her
bilateral interactions with each of the $n-1$ other players. Let $S_i$ be player $i$'s set of pure strategies. The game is
represented by the payoff matrices $A^{ij}\in \Real^{|S_i|\times| S_{j}|}$ for all pairs of players $(i,j)$.
Let $x^i \in \Delta(S_i)$ denote a mixed strategy of player $i$. 
Given a strategy profile $(x^1,\ldots,x^n)$, the expected utility of player $i$ is, 
\[
u_i(x^1,\ldots,x^n)=\sum_{j\neq i} (x^i)^T A^{ij} x^j
\]

We show that any polymatrix game can be transformed to a symmetric bilinear game 
such that any symmetric NE of the bilinear game corresponds to a NE of the polymatrix game.
Our derivation is adapted from Howson's \cite{howson1972equilibria} formulation of NE of polymatrix games as solutions
of an LCP. Formally, given a polymatrix game, we define the \emph{induced symmetric bilinear game} as $(A,A^T,E,E,e,e)$,
where
\[
A = \left[ \begin{array}{cccc}
0&A^{12}&\cdots&A^{1n}\\
A^{21}&0 & & A^{2n}\\
\vdots & & \ddots& \\
 A^{n 1}&A^{n2} & &0 \end{array}\right], \ \ \ 
E=\left[\begin{array}{cccc}
1_{|S_1|}^T & 0 & \cdots & 0\\
0 & 1_{|S_2|}^T & \cdots & 0\\
\vdots & & \ddots & \\
0 & 0 & \cdots & 1_{|S_n|}^T\end{array}\right],
\] 
and $e=1_n$. The space of strategy vectors is $x=[x^{1^T} \cdots x^{n^T}]^T$.

\begin{proposition}
Consider a polymatrix game of $n$ players.
The strategy $(x^1,x^2,\ldots,x^n)$ is an NE of the game if and only if $(x,x)$ is a symmetric NE of its induced bilinear
game.  
\end{proposition}
\begin{proof}
The proof is relatively straightforward, by observing that the respective incentive constraints are equivalent.
Thus the problem of finding a NE of a polymatrix game reduces to the problem of finding a symmetric NE of a symmetric bilinear game.
Note that an asymmetric NE doesn't correspond to a NE of the polymatrix game.
\qed
\end{proof}
\end{example}

Immorlica \emph{et al.} \cite{stoc11Dueling} analyzed several classes of  games between two optimization algorithms whose
objectives are to outperform each other. The space of pure strategies are the possible outputs of the algorithm which are
exponential, however the authors were able to formulate some of these games as zero-sum bilinear games (which they call
{\em bilinear duels}). We describe one example from
\cite{stoc11Dueling}.

\begin{example}{\em Ranking Duels \cite{stoc11Dueling}}\\
Each player chooses a ranking over $m$ elements. Thus the number of pure strategies is exponential in $m$. Such a ranking
can be represented as a $m\times m$ permutation matrix. By the Birkoff-von Neumann theorem, the space of mixed strategies
corresponds to the space of $m\times m$ \emph{doubly-stochastic} matrices with each row and each column sum to 1.
This space can be described by the polytope $\{x\in \Real^{m^2}\ |\  x\geq 0;\ \forall i,\ \sum_j x_{ij}=1;\ \forall  j,
\sum_i x_{ij}=1\}$.  Viewing $x$ and $y$ as column vectors, the sizes of the corresponding $E,e$ are polynomial in $m$.
Immorlica \emph{et al.} \cite{stoc11Dueling} constructed matrices $A\in \Real^{m^2\times m^2}$ such that the players' expected utilities are
equal to $x^TAy$ and $-x^TAy$ respectively. \qed
\end{example}

\begin{example}{\em Two-player Perfect-recall Extensive-form Games}\\
Extensive-form represents a dynamic game as a tree \cite{MAS}, where every pure strategy of a player prescribes a move at
each of the player's information sets. As a result the number of pure
strategies may be exponential in the size of the extensive-form description. 
Fortunately, if we assume perfect recall---roughly, that each player remembers all her
past decisions and observations---then there always exists a Nash equilibrium in \emph{behavior strategies}, where each
player independently chooses a distribution over actions at each of her information sets. Representation of a behavior
strategy requires space linear in the extensive form.  However, the expected utilities of the two-player perfect-recall
extensive-form games are not bilinear functions of the behavior strategies.  Koller \emph{et al.}
\cite{KollerMegiddoVonStengel1996}
proposed the \emph{sequence form}, which is a bilinear game formulation for these games.
The number of rows and columns of the matrices $A$ and $B$, in the bilinear form, are the number of feasible sequences of 
plays of the first and second player respectively. If a play sequence pair $(i,j)$ leads to a leaf node then the $ij^{th}$
entry of $A$ and $B$ are the payoffs of the first and second player at that leaf node, otherwise it is zero
\symbolfootnote[1]{Due to chance moves, the entry may correspond to multiple leaf nodes. In that case the entry stores the
expected payoff.}. A strategy $x$ of the first player is such that,
$x(\mbox{root})=1$, and if a sequence $\sigma$ ends at an information node $C$, then $\sum_{a \in Actions(C)} x_{\sigma a} -
x_\sigma=0$. Similar conditions hold for a strategy $y$ of the second player. Such a strategy may be transformed to
a behavior strategy and vice versa. This gives $E,F,e$ and $f$. Note that the reduction is
polynomial sized and $x^TAy$ and $x^TBy$ are exactly the expected payoffs under the corresponding behavior strategies. We
refer readers to \cite{KollerMegiddoVonStengel1996} for more details. \qed

\end{example}

\begin{remark}
Lemke's algorithm on the LCP formulation of bilinear games terminates with a solution
({\em i.e.}, not at a ray) if the only non-negative solutions $x$ and $y$ to $Ex = 0$ and $Fy = 0$ are $x = 0$ and
$y = 0$, and the payoff matrices are non-positive, {\em i.e.}, $A \le 0$ and $B \le 0$. This result directly follows from
\cite{KollerMegiddoVonStengel1996}. Note that games of all the above examples satisfy these requirements, without loss of generality.
\end{remark}

The rank of a game $(A,B)$ is defined as $rank(A+B)$, and we consider the rank based hierarchy of the bilinear games. The
set of rank-$k$ games consists of all $(A,B)$ such that $rank(A+B)\le k$. Zero-sum games are rank-$0$ games, the smallest
set in the hierarchy.
Koller et al. \cite{koller1994} gave an LP formulation for zero-sum bilinear games, derived from two-player extensive
form games with perfect recall. However, their formulation works for general bilinear games as well.
Beyond rank-$0$ games, no polynomial time algorithm is known for NE computation (even for the reduction
specific formulations). In the next section, we extend the polynomial time solvability of Nash equilibrium for the rank-$1$
bilinear games. 

For all the algorithms that follow, we make the following assumptions (without loss of generality): {\em 1)} The
entries of $A,B,E,F,e$ and $f$ are integers, since scaling them by a positive
number does not change the set of NE. {\em 2)} The equalities $Ex=e$ and $Fy=f$ are all linearly independent because even if we
discard the dependent equalities, $X$ and $Y$ do not change. {\em 3)} The letter $\mathcal L$ denotes the bit length of the input game.

\section{Rank-1 Games and Polynomial Time Algorithm}\label{rank1}
The approach used in this section is motivated by the paper \cite{agms}.
Given a rank-$1$ game $(A,B)$, it is easy to find $\alpha \in \Real^M$ and $\beta \in \Real^N$ such that $A+B=\alpha \cdot \beta^T$, since
any two rows of $A+B$ are multiple of each other. In that case $B=-A+\alpha \cdot \beta^T$. Let $G(\alpha)=(A, -A+\alpha \cdot \beta^T)$ be
a parametrized game for a fixed $A \in \Real^{M\times N}$ and $\beta \in \Real^N$. For any game $G(\alpha)$ the BRP of 
first-player is fixed to $P(\alpha)=P$ (of (\ref{eq2})) since $A$ is fixed. However, the BRP of second-player $Q(\alpha)$
changes with the parameter. Now, consider the following polytope with $x$, $q$ as vector variables and $\lambda$ as a scalar
variable:
\begin{eqnarray}\label{eq5}
\begin{array}{llclcll}
Q'=\{&(x,\lambda,q)\in \Real^{M+1+k_2}\ |& x_i\geq 0,& \forall i \in S_1;&\hspace{.06in}
x^T(-A^j)+\lambda\beta_j-q^TF^j\leq 0,& \forall j \in S_2;& \hspace{.06in}Ex=e\}  
\end{array}
\end{eqnarray}

It is easy to see that $Q(\alpha)$ is the projection of $\{(x,\lambda,q)\in Q'\ |\ \lambda=x^T\alpha\}$ on $(x,q)$-space. In other words,
$Q(\alpha)$ is a section of $Q'$ obtained by hyper-plane $\lambda=x^T\alpha$. Clearly, $Q'$ covers $Q(\alpha), \forall \alpha \in
\Real^M$. Number the equations of $Q'$ in a similar way as the equations of $Q$.
Let $\mathcal N$ be the set of fully-labeled pairs of $P\times Q'$, {\em i.e.,} $\mathcal N=\{(v,w) \in P\times Q'\ |\ L(v) \cup
L(w)=\{1,\dots,M+N\}\}$. Using the definition of fully-labeled pairs, it is easy to check that for a given $((y,p),(x,\lambda,q))\in P\times Q'$, 
\begin{eqnarray}\label{eq45}
((y,p),(x,\lambda,q))\in \mathcal N\ \ \Leftrightarrow\ \ x^T(Ay - E^T p)=0\ \mbox{  and  }\ (x^T(-A) +\lambda \beta^T-q^TF)y=0
\end{eqnarray}

\begin{lemma}\label{le1}
Let $(v,w) \in \mathcal N$, $v=(y,p)$ and $w=(x,\lambda,q)$.
\begin{itemize}
\item For all $\alpha$ such that $\lambda=x^T\alpha$, $(x,y)$ is a NE of $G(\alpha)$.
\item For every NE $(x,y)$ of a game $G(\alpha)$, there exists a $(v,w) \in \mathcal N$, where $\lambda=x^T\alpha$.
\end{itemize}
\end{lemma}
\begin{proof}
Since $(v,w)$ is fully-labeled it satisfies, $x^T(Ay - E^T p)=0$ and $(x^T(-A) +\lambda \beta^T-q^TF)y=0$. Let $\alpha$ be
such that $\lambda=x^T\alpha$ then we get $(x^T(-A) +(x^T\alpha) \beta^T-q^TF)y=0 \Rightarrow
(x^T(-A+\alpha\beta^T)-q^TF)y=0$. This implies that $(x,y)$ is a NE of the game $(A,-A+\alpha\beta^T)$ ({\em i.e.,}
$G(\alpha)$), since it satisfies the complementarity condition of (\ref{eq3}).

Given a $(x,y)$ of $G(\alpha)$, from (\ref{eq1}) and (\ref{eq3}) it is clear that $\exists p,q$ such that $x^T(Ay-E^Tp)=0$ and
$(x^T(-A+\alpha\beta^T)-q^TF)y=0$. Let $\lambda=x^T\alpha$, then we get $(x^T(-A) +\lambda \beta^T-q^TF)y=0$. Therefore,
$((y,p),(x,\lambda,q)) \in \mathcal N$. \qed
\end{proof}

The above lemma establishes strong relation between the set of NE of all the $G(\alpha)$s and the set $\mathcal N$. Next we discuss the structure
of $\mathcal N$, and later use it to design a polynomial time algorithm to find a NE of a given game $G(\alpha)$.

The polytopes $P$ and $Q'$ are assumed to be non-degenerate\footnote{Degeneracy may be handled using standard techniques as done in \cite{agms}.}, and let
$k_1=k_2=k$ for simplicity.  As there are $k$ linearly independent equalities in $P$ and $Q'$, they are of dimension $N$ and $M+1$ respectively.
Therefore, $\forall (v,w)\in \mathcal N$, $|L(v)|\le N$ and $|L(w)| \leq M+1$. Since, $M+N$ labels are required for a pair
$(v,w)$ to be part of $\mathcal N$, $\mathcal N \subset 1$-skeleton of $P\times Q'$. Further, if $(v,w) \in \mathcal N$ is a vertex pair
then $|L(v) \cap L(w)|=1$. Let the label in the intersection be called the {\em duplicate label} of $(v,w)$. Relaxing the
inequality corresponding to the duplicate label at $(v,w)$ in $P$ and $Q'$ respectively gives its two adjacent edges in
$\mathcal N$. Therefore, every vertex of $\mathcal N$ has degree two. This implies that $\mathcal N$ is a set of cycles and
infinite paths (unbounded edges at both the ends). We will show that
$\mathcal N$ forms a single infinite path.  The next lemma follows directly from the definition of $P$ (\ref{eq2}) and $Q'$
(\ref{eq5}), and expression (\ref{eq45}).

\begin{lemma}\label{le2}
For all $(v, w) = ((y,p), (x,\lambda,q)) \in P\times Q'$, we have $\lambda(\beta^Ty)-e^Tp-f^Tq\leq0$, and the equality holds iff $(v,w)\in
\mathcal N$.
\end{lemma}

Lemma \ref{le2} implies that $\mathcal N$ is captured by $\lambda(\beta^Ty)-e^Tp-f^Tq=0$ over $P\times Q'$. Using this fact
and Lemma \ref{le2}, we define the following parametrized LP.

\[
\begin{array}{rl}
LP(\delta)\ -\ \mbox{max:} & \delta(\beta^Ty)-e^Tp-f^Tq \\
&\mbox{s.t. } \ \ P \times Q';\ \ \lambda=\delta
\end{array}
\]

For an $a \in \Real$, let $OPT(a)$ be the set of optimal solutions of $LP(a)$ and $\mathcal N(a)$ be the set of points of $\mathcal N$ with
$\lambda=a$, {\em i.e.,} $\mathcal N=\{(v,w) \in \mathcal N\ |\ w=(x,\lambda,q)\mbox{ and } \lambda=a\}$.

\begin{lemma}\label{le3}
For an $a \in \Real$, $\mathcal N(a)\neq \emptyset$ and $OPT(a)=\mathcal N(a)$
\end{lemma}
\begin{proof}
Consider a game $G(\alpha)$ where $\alpha=[a,\dots,a]$. Clearly, for any Nash equilibrium $(x,y)$ of $G(\alpha)$ the
corresponding point in $\mathcal N$ has $\lambda=x^T\alpha=a$ (Lemma \ref{le1}). Therefore, $\mathcal N(a)\neq \emptyset$.
The feasible set of $LP(a)$ is all points of $P\times Q'$ with $\lambda=a$. Further, function 
$\lambda(\beta^Ty)-e^Tp-f^Tq$ achieves maximum only at points on $\mathcal N$ (Lemma \ref{le2}).
Therefore, $OPT(a)=\mathcal N(a)$ follows. \qed
\end{proof}

\begin{lemma}\label{le4}
The set $\mathcal N$ forms an infinite path, with $\lambda$ being monotonic on it.
\end{lemma}
\begin{proof}
To the contrary suppose there are cycles and multiple paths in $\mathcal N$. Let $\mathcal C$ be a cycle in $\mathcal N$.
It is easy to see that $\mathcal N(a)$ = intersection of $\mathcal N$ with the hyper-plane $\lambda=a$. Therefore, $\exists a \in \Real$,
such that either $\mathcal C$ is contained in $\lambda=a$ or it cuts the cycle at exactly two points. This contradicts that $\mathcal N(a)$
is a convex set in both the cases (Lemma \ref{le3}).

Now let $\mathcal P_1$ and $\mathcal P_2$ be two paths in $\mathcal N$. Since, $\mathcal N(a)$ is a convex set $\forall a\in \Real$,
$\lambda$ is monotonic on both the paths. Suppose, the range of $\lambda$ covered by $\mathcal P_1$ and $\mathcal P_2$ be $(-\infty,a]$ and
$(a,\inf)$. However, this contradicts the fact that $\mathcal P_2$ is a closed set. 
Monotonicity of $\lambda$ follows from the convexity of $\mathcal N(a)$.
\qed
\end{proof}

\subsection{Algorithm}
Let $(A,B)$ be a given rank-$1$ game and $A+B=\gamma\cdot \beta^T$. Let $\gamma_{min}=\min_{x \in X} \sum_{i \in S_1} \gamma_ix_i$ and
$\gamma_{max}=\max_{x \in X} \sum_{i \in S_1} \gamma_ix_i$. The $\gamma_{min}$ and $\gamma_{max}$ exists since $X$ is a bounded polytope.
From Lemma \ref{le1} it is clear that every point in the intersection of the set
$\mathcal N$ and hyper-plane $H_\gamma: \lambda - \sum_{i \in S_1}\gamma_i x_i=0$ corresponds to a NE of the given game
$(A,B)$.  Note that for any point in the intersection, corresponding $\lambda$ is between $\gamma_{min}$ and $\gamma_{max}$. 
Let $H_\gamma^-$ and $H_\gamma^+$ be the negative
and positive half spaces of the hyper-plane $H_\gamma$ respectively, then clearly $\mathcal N(\gamma_{min})\in H_\gamma^-$ and  $\mathcal
N(\gamma_{max})\in H_\gamma^+$. All the points in the intersection of $\mathcal N$ and $H_\gamma$ are between $\mathcal N(\gamma_{min})$ and
$\mathcal N(\gamma_{max})$. The following algorithm does binary search on $\mathcal N$ between $\mathcal N(\gamma_{min})$
and $\mathcal N(\gamma_{max})$ to find a point in the intersection using the fact that $\lambda$ monotonically increases
(similar to the algorithm in \cite{agms}). 

\begin{itemize}
\item[$\bs{S_1}$]Initialize $a_1=\gamma_{min}$ and $a_2=\gamma_{max}$.
\item[$\bs{S_2}$]If the edge containing $\mathcal N(a_1)$ or $\mathcal N(a_2)$ intersects $H_\gamma$, then output the intersection
and exit. 
\item[$\bs{S_3}$]Let $a=\frac{a_1+a_2}{2}$. Let $\overline{u,v}$ be the edge containing $\mathcal N(a)$.
\item[$\bs{S_4}$]If $\overline{u,v}$ intersects $H_\gamma$, then output the intersection and exit.
\item[$\bs{S_5}$]Else if $\overline{u,v} \in H^-_\gamma$, then set $a_1=a$ else set $a_2=a$ and continue from step $S_3$.
\end{itemize}

\noindent{\bf Correctness.}
Since the feasible set of $LP(a)$ is a section of $P\times Q'$, where $\lambda=a$, the $OPT(a)$ is on an edge of
$P\times Q'$ (assuming non-degeneracy of $LP(a)$). It is easy to construct this edge from the tight equations of
$P\times Q'$ at $OPT(a)$. Clearly, this entire edge should be part of the set $\mathcal N$, hence if this edge intersects the
hyper-plane $H_\gamma$ then we get a Nash equilibrium of the given game. Since, all
the points in the intersection of $\mathcal N$ and $H_\gamma$ are between $\mathcal N(\gamma_{min})$ and $\mathcal N(\gamma_{max}$), and
$\lambda$ monotonically increases between these two (Lemma \ref{le4}), the algorithm does a simple binary search between $\gamma_{min}$
and $\gamma_{max}$ to find an $a$, such that the edge containing $\mathcal N(a)=OPT(a)$ intersects $H_\gamma$ (Lemma
\ref{le3}).  \\ \\
\noindent{\bf Time Complexity.}
Recall that $\mathcal L$ is the bit length of the input game. 
Since, $\gamma_{min}$ and $\gamma_{max}$ are optimal points of two LPs on set $X=\{Ex=e, x\ge 0\}$, they can be represented in
$poly(\mathcal L,M,N)$ bits. 
Let $Z=\max\{|A|,|E|,|F|,|e|,|f|,|\gamma|,|\beta|\}$, $l=M+N+k_1+k_2+1$, and $\Delta=l! Z^l$.

\begin{theorem}
The above algorithm finds a NE of game $(A,B)$ in time $poly(\mathcal L,M,N)$.
\end{theorem}
\begin{proof}
One round of steps $S_3$ to $S_5$ can be done in polynomial time since computation of $\mathcal N(a)$ requires solving
$LP(a)$ (Lemma \ref{le3}), and computation of $\overline{u,v} \cap H_\gamma$ requires checking the feasibility of a
polytope. Now, to show polynomial time complexity, we need to bound the number of rounds of steps $S_3$ to $S_5$.

Note that the denominator of any co-ordinate of a vertex of $P\times Q'$ is at most $\Delta$, 
and if $\lambda$ is not constant on an edge of $\mathcal N$, then the difference in its value between the two end points
of the edge is at least $\frac{1}{\Delta^2}$. Therefore, if $a_2-a_1 < \frac{1}{\Delta^2}$ the algorithm terminates. After $k$
rounds $a_2-a_1=\frac{\gamma_{max}-\gamma_{min}}{2^k}$. In round $k$ if $a_2-a_1=\frac{\gamma_{max}-\gamma_{min}}{2^k} >
\frac{1}{\Delta^2}$, then $k<\log(\gamma_{max}-\gamma_{min})+2\log\Delta$. Therefore, the algorithm is guaranteed to terminate after
$\log(\gamma_{max}-\gamma_{min})+2\log\Delta+1=poly(\mathcal L,M,N)$ many rounds. \qed
\end{proof}

\section{FPTAS for Rank-$k$ Games}\label{fptas}
In this section, we discuss fully polynomial time approximation schemes for fixed rank games ({\em i.e.,}
$rank(A+B)$ is constant). The approximation notion in bilinear games can be defined in a similar way to that of bimatrix
games given by Kannan et al. \cite{kannan09rank}.
Let $x_{max}=\max_{x \in X} \sum_i x_i$, $y_{max}=\max_{y \in Y} \sum_j y_j$ and $D=|A+B|$. 
Clearly the total payoff derived from a strategy profile $(x,y) \in X\times Y$ is at most $x_{max}Dy_{max}$. Using this we
define an $\epsilon$-approximate NE for a bilinear game $(A,B)$ as follows.

\begin{definition}\label{de1}
For a strategy profile $(x,y) \in X\times Y$, let $u=\max_{x' \in X} x'^TAy$ and $v=\max_{y' \in Y} x^TBy'$. Then $(x,y)$ is
an $\epsilon$-approximate NE of the game $(A,B)$ if $u+v-x^T(A+B)y\le\epsilon (x_{max}Dy_{max})$.
\end{definition}

For a bimatrix game $x_{max}Dy_{max}=D$, since $x$ and $y$ are probability distributions, which is compatible with the definition of
\cite{kannan09rank}. 
Next we define a stronger notion of $\epsilon$-approximate NE called {\em relative} $\epsilon$-approximate NE, where the error is relative
to the maximum achievable payoff from the given strategy.

\begin{definition}\label{de2}
For a strategy profile $(x,y)\in X\times Y$, let $u=\max_{x' \in X} x'^TAy$ and $v=\max_{y' \in Y} x^TBy'$. Then $(x,y)$ is
a {\em relative} $\epsilon$-approximate NE of the game $(A,B)$ if $u+v-x^T(A+B)y\le\epsilon (u+v)$, {\em i.e.,} the total error is relatively small.
\end{definition}

Since the value of $u+v$ is at most $x_{max}Dy_{max}$, if $(x,y)$ is relative $\epsilon$-approximate NE, then it is also
$\epsilon$-approximate NE. 
For all the examples mentioned in Section \ref{examples}, a (relative) approximate NE of the bilinear game formulation can be
straightforwardly turned into a (relative) approximate NE of the corresponding finite game under standard definitions.
Without loss of generality we assume that $A,B,E,F,e$ and $f$ are integer matrices, since scaling them by a positive value does not change
the set of (relative) $\epsilon$-approximate NE. 
Next we discuss two FPTAS to solve QP of (\ref{eq4}), one for each definition of
approximation. The approaches used in these algorithms are generalization of the ones in \cite{kannan09rank}. 

\subsection{FPTAS for Approximate NE}
We show that the result by Vavasis \cite{vavasis} can be applied to get an $\epsilon$-approximate Nash equilibrium
(Definition \ref{de1}).  The following proposition states the result by Vavasis.

\begin{proposition}\label{propvavasis}
Let $min\{\frac{1}{2}x^TQx+q^Tx:Ax\le b\}$ be a quadratic optimization problem with compact polytope $\{x\in \Real^n :
Ax\le b\}$, and let the rank of $Q$ be a fixed constant. If $x^*$ and $x^\#$ denote points minimizing and maximizing the
objective function $f(x)=\frac{1}{2}x^TQx+q^Tx$ in the feasible
region, respectively, then one can find in time $poly(\mathcal L, \frac{1}{\epsilon})$ a point $x^\diamond$ satisfying 
\[
f(x^\diamond)-f(x^*) \le \epsilon(f(x^\#) - f(x^*)).
\]
\end{proposition}
Now consider the following QP formulation of (\ref{eq4}), which captures all the NE of $(A,B)$ at its optimal.
\[
\begin{array}{rl}
\mbox{min:} & e^Tp+f^Tq-x^T(A+B)y \\
\mbox{s.t. } & Ay - E^Tp \le 0;\ \ Fy=f;\ \ y\ge 0 \\
& x^TB-q^TF \le 0;\ \ Ex=e;\ \ x\ge 0
\end{array} 
\]

\begin{theorem}
Let $(A,B)$ be a rank-$k$ game, then for every $\epsilon>0$, an $\epsilon$-approximate Nash equilibrium can be computed in
time $poly(\mathcal L,\frac{1}{\epsilon})$, where $\mathcal L$ is the bit length of the game and $k$ is a constant.
\end{theorem}
\begin{proof} 
The objective function of the above QP can be easily transformed to the standard QP form $\frac{1}{2}x^TQx+q^Tx$, where
$rank(Q)=2k$. To apply Proposition \ref{propvavasis} on this QP, we need to bound its feasible set. Since, $\{x : Ex=e,\ x\ge 0\}$
and $\{y : Fy=f,\ y\ge0\}$ are compact, the only variables to bound are $p$s and $q$s. Since, the maximum possible value
of $x^T(A+B)y$ for any $(x,y) \in X\times Y$ is $x_{max}Dy_{max}$, the value of $e^Tp+f^Tq$ is at most $x_{max}Dy_{max}$ at
any point of the polytope corresponding to NE (by (\ref{eq3})). Therefore, we
impose $e^Tp+f^Tq\le x_{max}Dy_{max}$. However, this may not bound the $p$s and $q$s. 

Let $Z=\max\{|A|,|B|,|E|,|F|,|e|,|f|\}$ and $l=M+N+k_1+k_2$. 
Recall that NE of a non-degenerate game correspond to vertices of the polytope. It is easy to see that maximum absolute value of a
co-ordinate of any vertex in the polytope is at most $l! Z^l$. Further, the quantity $l! Z^l$ can be represented in $poly(\mathcal L)$ bits.
Therefore, imposing $-l! Z^l \le p \le l!Z^l$ and $-l!Z^l \le q \le l!Z^l$ in the above QP incur only a polynomial increase
in its representation and does not change its optimal set. 
The minimum and the maximum objective values of this QP are zero (Lemma \ref{le2}) and at most $2x_{max}Dy_{max}$
respectively. 
Let $((y^\diamond, p^\diamond),(x^\diamond,q^\diamond))$ be the solution given by Vavasis algorithm for $\frac{\epsilon}{2}$, then 
from Proposition \ref{propvavasis} we get,
\[
e^Tp^\diamond + f^Tq^\diamond - {x^\diamond}^T (A+B) y^\diamond \le \epsilon (x_{max} D y_{max})
\]
From the primal-dual formulation of (\ref{pd}) it is clear that $\max_{x' \in X} x'^T A y^\diamond \le e^T p^\diamond$ and $\max_{y' \in Y}
{x^\diamond}^TBy' \le f^T q^\diamond$. Therefore, we get $\max_{x' \in X} x'^T A y^\diamond+max_{y' \in Y}
{x^\diamond}^TBy'- {x^\diamond}^T (A+B) y^\diamond \le \epsilon (x_{max} D y_{max})$. \qed
\end{proof}

\subsection{FPTAS for Relative Approximate NE}
Let the rank of a game $(A,B)$ be $k$, then $A+B=\sum_{i=1}^k \alpha(i)\beta(i)^T$, where $\forall i,\ \alpha(i) \in
\Real^M$ and $\beta(i) \in \Real^N$. We assume that the game is such that $\alpha(i)$s and $\beta(i)$s are positive vectors.
For all $i\le k$, let $w_i = \min_{x \in X} x^T\alpha(i)$ and $w'_i=\max_{x \in X} x^T\alpha(i)$, similarly let $z_i=\min_{y
\in Y} \beta(i)^Ty$ and $z'_i = \max_{y \in Y} \beta(i)^Ty$. 
Note that $w_i,w'_i,z_i$ and $z'_i$ can be represented by $poly(\mathcal L,M,N)$ bits, since $X$ and $Y$ are compact. 
Given an $\epsilon>0$, consider the sub-intervals $[w_i,(1+\epsilon)w_i]$, $[(1+\epsilon)w_i, (1+\epsilon)^2 w_i]$ of $[w_i, w'_i]$ and
similarly of $[z_i, z'_i]$. 
All combinations of these intervals form a grid in $2k$-dimensional box $\mathcal B=\times_i [w_i, w'_i]\times_i[z_i,
z'_i]$.  Let $(x,y)\in X\times Y$ be such that $\forall i, x^T\alpha(i) \in [u_i, (1+\epsilon)u_i]$ and $\beta(i)^Ty \in
[v_i, (1+\epsilon)v_i]$, then clearly, 
\begin{eqnarray}\label{eq55}
\sum_{i=1}^k u_i v_i \le x^T(A+B)y \le (1+\epsilon)^2 \sum_{i=1}^k u_i v_i
\end{eqnarray}

For a fixed hyper-cube of the grid, consider the following LP based on the QP of (\ref{eq4})
\[
\begin{array}{rl}
\mbox{min:} & e^Tp+f^Tq \\
\mbox{s.t. }& Ay\le  E^T p,\ \ Fy=f, \ \ y\ge 0 \\
& x^TB\le q^T F;\ \ Ex=e; \ \ x\ge 0 \\
& u_i \le x^T\alpha(i)\le (1+\epsilon) u_i;\ \ v_i \le \beta(i)^Ty\le (1+\epsilon) v_i,\ \ \forall i \\
\end{array}
\]

\noindent{\bf Algorithm.} Run the above LP for each hyper-cube of the grid, and output an optimal point of the one giving the best
approximation. 
As the number of hyper-cubes in the grid is $poly(\mathcal L,1/\log(1+\epsilon))$, the running time of the algorithm is
$poly(M,N,\mathcal L,1/\log(1+\epsilon))$. 
\vspace{0.2cm}

\noindent{\bf Correctness.}
Next we show that the above algorithm gives $\left(1-\frac{1}{(1+\epsilon)^2}\right)$-approximate NE of the game $(A,B)$. Let $(x',y')$ be a
NE of the given game, and $(p',q')$ be such that $x'^TAy'=e^Tp'$ and $x'^TBy'=f^Tq'$. Consider the hyper-cube containing
$(x'^T\alpha(1),\dots,x'^T\alpha(k),\beta(1)^Ty',\dots,\beta(k)^Ty')$ of the grid and corresponding LP. Clearly, $(x',y',p',q')$ is a
feasible point of this LP and
$\sum_{i=1}^k u_iv_i \le e^Tp' + f^Tq' \le (1+\epsilon)^2 \sum_{i=1}^k u_iv_i$, since
$e^Tp'+f^Tq'=x'^T(A+B)y'$. Therefore, at the optimal point $(\tilde{x},\tilde{y},\tilde{p},\tilde{q})$ of the LP we get
$e^T\tilde{p}+f^T\tilde{q} \le (1+\epsilon)^2 \sum_{i=1}^k u_iv_i$,
and this gives,

\[
\begin{array}{rl}
&\tilde{x}^T(A+B)\tilde{y} \ge \sum_{i=1}^k u_iv_i \ge \frac{e^T\tilde{p}+f^T\tilde{q}}{(1+\epsilon)^2} \ \ \ \ \ (\mbox{using }(\ref{eq55}))\\
\Rightarrow &e^T\tilde{p}+f^T\tilde{q} - \tilde{x}^T(A+B)\tilde{y} \le \left(1-\frac{1}{(1+\epsilon)^2}\right) (e^T\tilde{p}+f^T\tilde{q})\\
\end{array}
\]

Let $\mu=\left(1-\frac{1}{(1+\epsilon)^2}\right)$, $\tilde{u}=\max_{\gamma \in X} \gamma^TA\tilde{y}$, and $\tilde{v}=\max_{\gamma \in Y} \tilde{x}^TB\gamma$.
Clearly, $e^T\tilde{p}\ge \tilde{u}$ and $f^T\tilde{q} \ge \tilde{v}$ (using (\ref{pd})).  
Let $D=e^T\tilde{p}+f^T\tilde{q}-\tilde{u}-\tilde{v}$, then
$\tilde{u}+\tilde{v}-\tilde{x}^T(A+B)\tilde{y}=e^T\tilde{p}+f^T\tilde{q}-D-\tilde{x}^T(A+B)\tilde{y}\le \mu (e^T\tilde{p}+f^T\tilde{q}) - D
\le \mu (e^T\tilde{p}+f^T\tilde{q}-D)= \mu (\tilde{u}+\tilde{v})$, since $\mu \in (0,1)$. Therefore, 
$(\tilde{x},\tilde{y})$ is a relative $\mu$-approximate NE of the given game $(A,B)$ (Definition \ref{de2}). 

\begin{theorem}
Let $(A,B)$ be a rank-$k$ game, and $A+B=\sum_{i=1}^k \alpha(i)\beta(i)^T$, such that $\alpha(i)$s and $\beta(j)$s are positive vectors.
Then given an $\epsilon > 0$, a relative $\left(1-\frac{1}{(1+\epsilon)^2}\right)$-approximate NE can be computed in time
$poly(\mathcal L, 1/\log(1+\epsilon))$, 
where $\mathcal L$ is the input bit length. \qed
\end{theorem}
\vspace{-0.1cm}

For a symmetric game ($B=A^T, E=F, e=f$), an (relative) $\epsilon$-approximate symmetric NE can be defined as an (relative)
$\epsilon$-approximate NE with the same strategies, {\em i.e.,} $x=y$. It is easy to check that, if we use the QP
formulation of (\ref{symQP}) instead of (\ref{eq4}) in any of these algorithms, then the output strategy is an (relative) $\epsilon$-approximate
symmetric NE strategy. 

\section{Games with a Low Rank Matrix}\label{slowrank}
In this section we show that if rank of even one payoff matrix ($A$ or $B$) is constant, then Nash equilibrium computation
can be done in polynomial time. Recall the best response polytopes $P$ and $Q$ (\ref{eq2}) for the bilinear game
$(A,B)$.

\begin{lemma}\label{le_lcp}
Given a game $(A,B)$, there exists a vertex pair $((y,p),(x,q))\in P \times Q$ such that $(x,y)$ is a NE of $(A,B)$.
\end{lemma}
\begin{proof}
All the solutions of (\ref{eq3}) over $P\times Q$ are the NE of the game, and existence of a solution is guaranteed (Proposition
\ref{existence}). Suppose $v\in P\times Q$ is a NE of $(A,B)$. It is easy to check that the entire face, formed by the set of tight equations at $v$, is
solution, and it contains at least one vertex, as $P\times Q$ is bounded from one side.
\qed
\end{proof}

\begin{lemma}\label{poly_vert}
Let $k_1=k_2=k$ and $rank(A)=l$. The polytope $P$ has at most $O(N^{l+k})$ vertices.
\end{lemma}
\begin{proof}
From (\ref{eq2}), it is clear that $P$ is in $(N+k)$-dimensional Euclidean space, however $Fy=f$ gives $k$ linearly
independent equalities.  Therefore, $P$ is of dimension $N$, and at a vertex of $P$, $N$ linearly independent inequalities
must be tight. Since $A$ is of rank $l$, $rank([A$ -$E])\le l+k$.  Therefore, $\exists! S \subset S_1,\ |S|=l+k$ such that
$\forall i \in S_1\setminus S, A_iy - p^T E^i \le 0$ are not needed in defining the polytope $P$. 
At a vertex, if $d$ inequalities are tight from $S$ then rest $N-d$ must be of type $y_j=0$, hence for a fixed $D \subset
S,\ |D|=d$, there are at most ${N \choose N-d}={N\choose d}$ choices to form a vertex. Therefore, the total number of
vertices are at most $\sum_{i=0}^{l+k} {l+k \choose i}{N \choose i} \le 2^{l+k} N^{l+k}$. \qed
\end{proof}

Note that if we remove the assumption $k_1=k_2=k$ then the exponent of $N$ turns out to be a linear function of $l,k_1$ and
$k_2$. 
A similar proof can be worked out for $Q$. 

\begin{theorem}\label{lowrank}
If rank of either $A$ or $B$ is constant then a Nash equilibrium of a bilinear game $(A,B)$ can be computed in polynomial
time, assuming $k$ to be a constant.
\end{theorem}
\begin{proof}
Suppose $rank(A)=l$ (a constant) and $v=(y,p)\in P$ be a vertex. We can check in polynomial time whether $v$ corresponds to
a NE or not as follows.
Let $S_x=\{i \in S_1\ |\ A_iy-p^TE^i=0\}$ and $S_y=\{j \in S_2\ |\ y_j>0\}$. Consider
all $(x,q) \in \Real^{M+k}$
such that
\[Ex=e;\ \ 
\begin{array}{llll}
\forall j \in S_y,\ & x^TB^j-q^T F^j =0;  \hspace{.5cm}&\forall i \in S_x,\ & x_i \ge 0 \\
\forall j \notin S_y, & x^TB^j -q^T F^j \le 0; & \forall i \notin S_x, & x_i=0. 
\end{array}
\]
Every such $(x,q)$ lies in $Q$ and makes a fully-labeled pair with $v$, and hence forms a NE (Lemma \ref{ne_le}). 
Note that such an $(x,q)$ can be obtained in polynomial time by solving an LP. Now the proof follows from Lemmas
\ref{le_lcp} and \ref{poly_vert}. A similar argument proves the other case when $rank(B)$ is constant. \qed

\end{proof}

As the set of bimatrix games is a subclass of the bilinear games (Example \ref{bimatrix}), where $k_1=k_2=1$, Theorem
\ref{lowrank} strengthens the results by Lipton, Markakis and Mehta \cite{lipton2003simple} (Corollary $4$), and Kannan and
Theobald \cite{kannan09rank} (Theorem $3.2$), where they require that the rank of both $A$ and $B$ to be constants. 
Note that $k$ in Bayesian games depends on the number of types of players and in the sequence form, it depends on the number
of information sets of players. Therefore, this result can be applied to these games if in their bilinear representation, a
payoff matrix has low rank and $k$ is constant. 

In fact Theorem \ref{lowrank} gives a polynomial time algorithm to enumerate all the extreme equilibria of a bilinear game
with a constant rank matrix, and an exponential time enumeration algorithm for any bilinear game. A similar (exponential
time) algorithm was given by Avis et al. \cite{avis} to enumerate all Nash equilibria of a bimatrix game.

\section{Conclusion}
We have defined two-player {\em bilinear games}, where payoffs are represented by two matrices $(A,B)$ and strategy sets are
compact polytopes. 
In both bilinear and bimatrix games, the utilities are bilinear functions of strategy vectors.
The scope of these games is large enough to capture many interesting classes of games like bimatrix games,
two-player Bayesian games, polymatrix games, and two-player extensive-form games with perfect recall. 
Considering the rank-based hierarchy puts a structure on bilinear games, and by exploiting this structure and the similarity
between bilinear and bimatrix games, we extended various combinatorial
and algorithmic results, pertaining to the efficient computation of Nash equilibria, from bimatrix games to bilinear games.
It will be interesting to know what other results of bimatrix games extend to bilinear games like {\em 1)} designing
Lemke-Howson type algorithm for NE computation, {\em 2)} extending other algorithms for computation of approximate
equilibria, etc.\\

\noindent{\bf Acknowledgments. }We would like to thank Milind Sohoni for helpful comments and corrections. 

\bibliographystyle{splncs03}
\bibliography{Albert,sigproc}
\end{document}